%% file: group4.tex
\documentclass{llncs}

\usepackage{amssymb}
\usepackage{amsmath}
\usepackage{epsfig}
\usepackage{float}

\def\Zee{{\mathbb Z}}

\newcommand{\vdashstar}{\mathrel{\mbox{$\vdash\hspace{-0.75em}{\raisebox{1.1ex}{$\scriptstyle*$}}$}}}

\newcommand{\be}{\begin{enumerate}}
\newcommand{\ee}{\end{enumerate}}

\begin{document}
\title{Automata and Reduced Words in the Free Group}
\author{Thomas Ang%
	\inst{1}
\and Giovanni Pighizzini%
	\inst{2}
\and Narad Rampersad%
	\inst{3}
\and Jeffrey Shallit%
	\inst{1}
}
\institute{David R.\ Cheriton School of Computer Science\\
	University of Waterloo
	Waterloo, ON, Canada N2L 3G1\\
	\email{\{tang,shallit\}@uwaterloo.ca}
\and Dipartimento di Informatica e Comunicazione\\
	Universit\`{a} degli Studi di Milano,
	via Comelico 39, 20135 Milano, Italy\\
	\email{pighizzini@dico.unimi.it}
\and Department of Mathematics and Statistics\\
	University of Winnipeg,
	Winnipeg, MB, Canada R3B 2E9\\
	\email{n.rampersad@uwinnipeg.ca}
}
\toctitle{Automata and Reduced Words in the Free Group}%
\tocauthor{T.~Ang, G.~Pighizzini, N.~Rampersad, J.~Shallit}%
\maketitle

\setcounter{footnote}{0} 

\begin{abstract}
We consider some questions about formal languages that arise when
inverses of letters, words and languages are defined. The reduced
representation of a language over the free monoid is its unique
equivalent representation in the free group. We show that the class of
regular languages is closed under taking the reduced representation,
while the class of context-free languages is not. We also give an upper
bound on the state complexity of the reduced representation of a
regular language, and prove upper and lower bounds on the length of the
shortest reducible string in a regular language. Finally we show that
the set of all words which are equivalent to the words in a regular 
language can be nonregular, and that regular languages are not closed 
under taking a generalized form of the reduced representation.
\end{abstract}

\section{Introduction}
A word in a free group can be represented in many different ways. For example, $a a a^{-1}$ and $a a^{-1} b a a^{-1} b^{-1} a$ are two different ways to write the word $a$. Among all the different representations, however, there is one containing no occurrences of a letter next to its own inverse. We call such a word {\it reduced}. In this paper we consider some basic questions about formal languages and their reduced representations. In previous work on automatic groups, these notions of inverse symbols and reduced words have been studied, but only with regards to automata that are assumed to generate groups \cite{Eps}.

First we define some standard notation. A deterministic finite automaton (DFA) is denoted by a quintuple $(Q, \Sigma, \delta, q_0, F)$ where $Q$ is the finite set of states, $\Sigma$ is the finite input alphabet, $\delta:Q\times\Sigma\rightarrow Q$ is the transition function, $q_0 \in Q$ is the initial state, and $F \subseteq Q$ is the set of accepting states. 
We generalize $\delta$ to the usual extended transition function 
with domain $Q \times \Sigma^*$. 
We use similar notations to represent a nondeterministic finite automaton with $\epsilon$-transitions ($\epsilon$-NFA), except the transition function is $\delta:Q\times(\Sigma\cup\{\epsilon\})\rightarrow 2^Q$. In an $\epsilon$-NFA it is possible to have $\epsilon$-transitions, which are transitions that can be taken without reading input symbols.
For a DFA or $\epsilon$-NFA $M$, $L(M)$ is the language accepted by
$M$. For any $x \in \Sigma^*$, $|x|$ denotes the length of $x$, and
$|x|_a$ for some $a \in \Sigma$ denotes the number of occurrences of
$a$ in $x$. We let $|\Sigma|$ denote the alphabet size. We use the
terms {\em prefix} and {\em factor} in the following
way. 
If there exist $x,z \in
\Sigma^*$ and $w= xyz$, we say that $y$ is a factor of $w$. If $x =
\epsilon$, we also say that $y$ is a prefix of $w$. If $y$ is a
factor or prefix of $w$ and $y \neq w$, then $y$ is a {\em proper\/}
factor or a proper prefix of $w$, respectively.  

In addition to this standard notation, we also define some notation
specific to our problem. For a letter $a$, we denote its inverse by
$a^{-1}$, and we let the empty word, $\epsilon$, be the identity. We
consider only alphabets of the form $\Sigma = \Gamma \cup \Gamma^{-1}$,
where $\Gamma = \{ 1, 2, \ldots\}$ and $\Gamma^{-1} = \{ 1^{-1},
2^{-1}, \ldots\}$. For a word $w\in\Sigma^* = a_1a_2\cdots a_n$, we
denote its inverse by $w^{-1} = a_n^{-1}\cdots a_2^{-1}a_1^{-1}$, and
for a language $L\subseteq \Sigma^*$, we let $L^{-1} = \{w^{-1} : w \in L\}$.
Note that taking the inverse of a word is equivalent to
reversing it and then applying a homomorphism that maps each letter to
its inverse. Now, we introduce a reduction operation on words, consisting
of removing factors of the form $aa^{-1}$, with $a\in\Sigma$.
More formally, let us define the relation $\vdash{}\subseteq\Sigma^*\times\Sigma^*$
such that, for all $w,w'\in\Sigma^*$, $w\vdash w'$ if and only if
there exists $x,y\in\Sigma^*$ and $a\in\Sigma$ satisfying $w=xaa^{-1}y$ and $w'=xy$.
As usual, $\vdashstar$ denotes the reflexive and transitive closure of $\vdash$.

\begin{lemma}\label{lem:red}
	For each $w\in\Sigma^*$ there exists exactly one word $r(w)\in\Sigma^*$
	such that $w\vdashstar r(w)$ and $r(w)$ does not contain any factor of the
	form $aa^{-1}$, with $a\in\Sigma$.
\end{lemma}
\begin{proof}
	First, we prove that if $w\vdash w'$ and $w\vdash w''$ then there
	exists $u\in\Sigma^*$ such that $w'\vdashstar u$ and $w''\vdashstar u$, i.e.,
	in the terminology of rewriting systems, $\Sigma^*$ with the relation $
	\vdash$ is a local confluent system.
	
	To this aim, suppose that $w=x'aa^{-1}y'=x''bb^{-1}y''$, $w'=x'y'$, $w''=x''y''$, 
	for some $x',x'',y',y''\in\Sigma^*$, $a,b\in\Sigma^*$, and, without loss
	of generality, that $|x'|\leq|x''|$.
	If $|x'|=|x''|$ then $w'=w''$, hence, we can take $u=w'$.
	Otherwise, if $|x'aa^{-1}|\leq|x''|$ then $w=x'aa^{-1}zbb^{-1}y''$, for
	some $z\in\Sigma^*$ and the desired word is $u=x'zy''$.
	In the only remaining case, $x'a=x''$, which implies that
	$b=a^{-1}$ and $w=x'aa^{-1}ay''$. Hence $w'=w''=x'ay''$, so we just take $u=w'$.
	
	Since $w\vdash w'$ implies that $|w'|=|w|-2$, no infinite
	reduction sequences are possible. By Newman's lemma~\cite{New}, this
	implies that the system is confluent, namely, for each $w\in\Sigma^*$ there
	exists exactly one word $r(w)$ such that $w\vdashstar r(w)$ and,
	for each $x\in\Sigma^*$, $r(w)\vdashstar x$ implies $x=r(w)$, i.e.,
	$r(w)$ does not contain any factor $aa^{-1}$.
\qed
\end{proof}
We define the \emph{reduced representation} of a word $w\in\Sigma^*$ as the
word $r(w)$ given in Lemma~\ref{lem:red}, i.e, the word which is obtained
from $w$ by repeatedly replacing with $\epsilon$ all factors 
of the form $aa^{-1}$, for any
letter $a\in\Sigma$, until no such factor exists. If
$r(w) = \epsilon$ we say that $w$ is \emph{reducible}. We can extend this
to languages so that for $L \subseteq \Sigma^*$, we let 
$r(L) = \{r(w): w\in L\}$.

Given a language $L$, the $\vdash$-closure of $L$ is the set of the
words which can be obtained by ``reducing'' words of $L$, i.e., the set
$\{x\in\Sigma^*:\exists w\in L\text{ s.t.\ } w\vdashstar x\}$.
Notice that $r(L)$ coincides with the intersection of the $\vdash$-closure of $L$
and $r(\Sigma^*)$.

Section 2 examines the reduced representations of regular and
context-free languages. Section 3 provides some bounds on the state
complexity of reduced representations. In Section 4 we look at bounds
on the length of the shortest word in a regular language that reduces
to $\epsilon$. Section 5 demonstrates counterexamples for some other
natural questions.

\section{Closure of Reduced Representation}
When considering the reduced representations of languages, it is
natural to wonder if common classes of languages are closed under this
operation. In this section we show that if a language $L$ is regular
then $r(L)$ is regular, but if $L$ is context-free then $r(L)$ does not
need to be context-free.

\begin{lemma}
\label{lem:sigstar}
For $\Sigma = \Gamma \cup \Gamma^{-1}$, where $\Gamma = \{ 1, \ldots, k\}$ and $\Gamma^{-1} = \{ 1^{-1}, \ldots, k^{-1}\}$, there exists a DFA
$M_k$ of $2k +2$ states that accepts $r(\Sigma^*)$.
\end{lemma}

\begin{proof}
We notice that a word $w$ is reduced if and only if it does not contain
the factor $aa^{-1}$, for each $a\in\Sigma$.
This condition can be verified by defining an automaton $M_k$
that remembers in its finite control the last input letter.
To this aim, the automaton has a state $q_a$ for each $a\in\Sigma$.
If in the state $q_a$ the symbol $a^{-1}$ is received, then the
automaton reaches a dead state $q_{-1}$.

Formally, $M_k$ is the DFA $(Q, \Sigma, \delta, q_0, F)$ 
defined as follows (see Figure~\ref{fig:sigstar} for an example):
$Q = \{q_0, q_{-1}\} \cup \{q_i : i
\in \Gamma \cup \Gamma^{-1}\}$, $F = Q\setminus \lbrace q_{-1}\rbrace$, and
\begin{equation}
\label{eq:deltasigstar}
\delta(q_a, c) =
\begin{cases}
q_c, & \text{if }a \neq c^{-1} \text{ and } q_a \neq q_{-1}; \\
q_{-1}, & \text{otherwise}.
\end{cases} \nonumber
\end{equation}
\begin{figure}
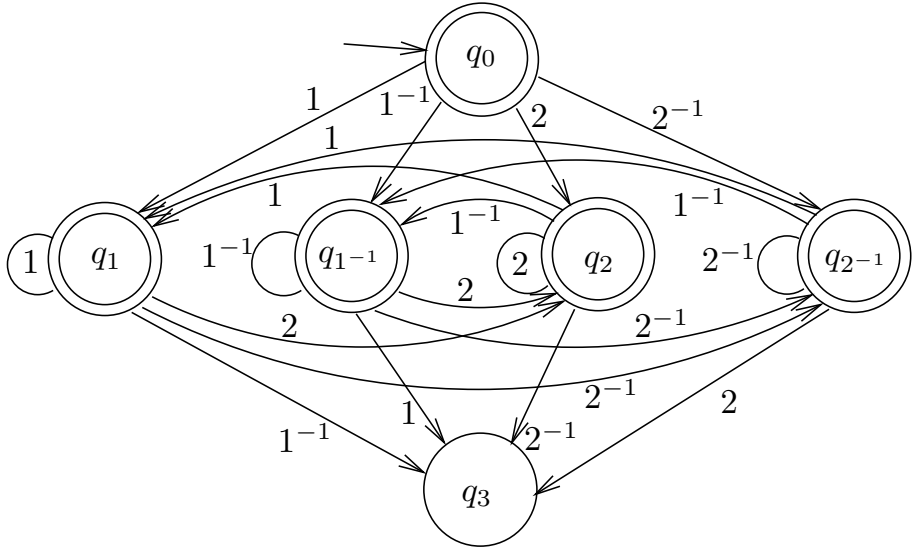

\begin{center}
\resizebox{\columnwidth}{!}{\input rsigstar.pstex_t}
\caption{
	$M_2$, a DFA that accepts $r(\Sigma^*)$ for $|\Sigma| = 4$. 
}
\label{fig:sigstar}
\end{center}
\end{figure}
\qed
\end{proof}

\begin{lemma}
\label{lem:clos}
Given an $\epsilon$-NFA $M = (Q, \Sigma, \delta, q_0, F)$ with $n$ states,
an automaton $M'$ accepting the $\vdash$-closure of $L(M)$
can be built in time $O(n^4)$.
\end{lemma}
\begin{proof}
The idea behind the proof is to present an algorithm that given $M$ 
computes $M'$ by adding to $M$ $\epsilon$-transitions corresponding to paths on reducible words. 
The algorithm is similar to a well known algorithm for minimizing DFAs \cite[p.\ 70]{HopUll}.
It uses a directed graph $G=(Q,E)$ to remember $\epsilon$-transitions.
For each pair of states $s,t$, the algorithm also keep a set
$l(s,t)$ of pairs of states, with the following meaning: if $(p,q)\in l(s,t)$ and the algorithm discovers that there is a path from $s$ to $t$ on
a reducible word (and hence it adds the edge $(s,t)$ to $G$),
then there exists a path from $p$ to $q$ on a reducible word (thus,
the algorithm can also add the edge $(p,q)$).
{
\begin{tabbing}
$E\leftarrow$ transitive closure of $\{(p,q)\mid q\in\delta(p,\epsilon)\}$\\
{\bf for} \= $s,t\in Q$ {\bf do} $l(s,t)\leftarrow\emptyset$\\
{\bf for} \= $p,q,s,t\in Q$ {\bf do}\\
	\>{\bf if} \=$\exists a\in\Sigma$ s.t.\ $s\in\delta(p,a)$ and
	$q\in\delta(t,a^{-1})$ {\bf then}\\
	\>\>{\bf if} $(s,t)\in E$ \= {\bf then} \= $E\leftarrow update(E,(p,q))$\\
	\>\>\>{\bf else} \> $l(s,t)\leftarrow l(s,t)\cup\{(p,q)\}$
\end{tabbing}
}
\noindent
The subroutine $update$ returns the smallest set $E'$ having the following
properties:
\begin{itemize}
\item $E\cup\{(p,q)\}\subseteq E'$;
\item if $(p',q')\in E'$ then each element belonging to $l(p',q')$ is in $E'$;
\item the graph $(Q,E')$ is transitive.
\end{itemize}
At the end of the execution, the automaton $M'$ is obtained by adding to $M$ an 
$\epsilon$-transition from a state $p$ to a state $q$ for each edge $(p,q)$ in the 
resulting graph~$G$.

\smallskip
Now, we show that the language accepted by $M'$ is the $\vdash$-closure of $L(M)$.
To this aim, we observe that for all $p,q,s,t\in Q$, such that $s\in\delta(p,a)$ and
$q\in\delta(t,a^{-1})$ are transitions of $M$ for some $a\in\Sigma$, 
if $M'$ contains an $\epsilon$-transition
from $s$ to $t$ then it must contain also an $\epsilon$-transition from $p$ to $q$.
In fact, when the algorithm examines these 4 states in the loop, if $(s,t)$ is in $E$ then
the algorithm calls $update$ to add $(p,q)$ to $E$. Otherwise, the algorithm adds
$(p,q)$ to $l(s,t)$. Since $M'$ finally contains the $\epsilon$-transition
from $s$ to $t$, then there is a step of the algorithm, after the insertion of 
$(p,q)$ in $l(s,t)$, adding the pair $(s,t)$ to $E$. 
The only part of the algorithm able to perform this operation is the
subroutine $update$.\footnote{Notice that the pair $(s,t)$ can be
added to $E$ by the subroutine $update$ either because it is the
second argument in the call of the subroutine, or because it belongs
to a list $l(p',q')$, where $(p',q')$ is added to $E$ in the same
call, or because there is a path from $s$ to $t$ consisting of some
arcs already in $E$ and at least one arc added during the same call of $update$.}
But when the subroutine adds the pair $(s,t)$
to $E$ then it must add all the pairs in $l(s,t)$. Hence,
$M'$ must also contain an $\epsilon$-transition from $p$ to $q$.
As a consequence, if $w\in L(M')$ and $w\vdash w'$ then $w'\in L(M')$, i.e.,
$L(M')$ is closed under $\vdash$.
Since the algorithm does not remove the original 
transitions from $M$, $L(M)\subseteq L(M')$ and, hence, the $\vdash$-closure
of $L(M)$ is included in $L(M')$.

On the other hand, it can be easily shown
that for each $\epsilon$-transition of $M'$ from a state $p$ to a state
$q$ there exists a reducible word $z$ such that $q\in\delta(p,z)$ in $M$.
Using this argument, from each word $w\in L(M')$ we can find a word
$x\in L(M)$ such that $x\vdashstar w$. This permit us to conclude that
$L(M')$ accepts the $\vdash$-closure of $L(M)$.

\smallskip
Now we show that the algorithm works in $O(n^4)$ time. A naive analysis gives a running time growing at least as $n^4$. In fact, the second for-loop
iterates over all state 4-tuples. Inside the loop the most expensive
step is the subroutine $update$. This subroutine starts by adding an
edge $(p,q)$ to $E$. For each new edge $(p',q')$ added to $E$ the subroutine
has to add all the edges in $l(p',q')$, while keeping the graph transitive.
This seems to be an expensive part of the computation. However, we can
observe that each set $l(p',q')$ contains less than $n^2$ elements. 
Furthermore, a set $l(p',q')$ is examined only once during the
execution of the algorithm, namely when $(p',q')$ is added to $E$. Hence, the
total time spent while examining the sets $l$ in all the calls of
the subroutine $update$ is $O(n^4)$. Furthermore, no more that $n^2$ edges can be inserted into $G$, and each insertion can be done in $O(n)$ amortized time while maintaining the transitive closure \cite{Ital,Pou}.  
Summing up, we get that the overall time of the algorithm is $O(n^4)$.
\qed
\end{proof}

By combining the results in the previous lemmata, we are now able to
show the following:

\begin{proposition}
\label{prop:regalg}
Given an $\epsilon$-NFA $M = (Q, \Sigma, \delta, q_0, F)$ with $n$ states,  an $\epsilon$-NFA $M_r$ such that $L(M_r) = r(L(M))$ can be built in $O(n^4)$ time.
\end{proposition}

\begin{proof}
The language $r(L(M))$ is the intersection of the  
$\vdash$-closure of $L(M)$ and $r(\Sigma^*)$.
According to Lemma~\ref{lem:clos}, from $M$ we
build an automaton $M'$ accepting the $\vdash$-closure of $L(M)$.
Hence, using standard constructions, from $M'$ and the automaton
obtained in Lemma~\ref{lem:sigstar} (whose size is fixed,
if the input alphabet is fixed), we get the automaton
$M_r$ accepting $r(M)=L(M')\cap r(\Sigma^*)$. 

The most expensive part is the construction of $M'$, which
uses $O(n^4)$ time.
\qed
\end{proof}

\begin{corollary}
\label{cor:regclosure}
For any $L \subseteq \Sigma^*$, if $L$ is regular then $r(L)$ is regular.
\end{corollary}

Now we turn our attention to context-free languages and prove that the analogue of Corollary \ref{cor:regclosure} does not hold. For this we use the notion of a quotient of two languages.

\begin{definition}
Given $L_1, L_2 \subseteq \Sigma^*$, the quotient of $L_1$ by $L_2$ is $$L_1/L_2 = \{w : \exists x \in L_2 \ {\rm such\ that } \ wx\in L_1\}$$
\end{definition}
While the class of regular languages is closed under quotients,
the class of context-free languages is not closed under this 
operation~\cite{Gins}.
It turns out that the reduced representation of a language can be used to compute quotients.
\begin{lemma}
\label{lem:quotient}
For any two languages $L_1,L_2\subseteq\Gamma^*$, the language $L_3 = r(L_1 L_2^{-1}) \cap \Gamma^*$ equals the quotient $L_1/L_2$.
\end{lemma}

\begin{proof}
We notice that $r(wxx^{-1}) = w$, for each $w,x\in\Gamma^*$.
Hence, given $w\in\Gamma^*$, it holds that $w\in L_3 = r(L_1 L_2^{-1}) \cap \Gamma^*$ if and only if there exists $x\in L_2$ such that $wx\in L_1$. Therefore $L_3 = L_1/L_2$.
\qed
\end{proof}

\begin{corollary}
\label{prop:cfclosure}
The class of context-free languages is not closed under $r()$.
\end{corollary}

\begin{proof}
By contradiction, suppose that the class of context-free languages
is closed unded $r()$.
Since this class is closed under the operations of
reversal, morphism, concatenation and intersection with a regular
language, for any two context-free languages $L_1$ and $L_2$ over
$\Gamma$, the language $L_3 = r(L_1 L_2^{-1}) \cap \Gamma^*$ is also
context-free. However, from Lemma \ref{lem:quotient}, $L_3 = L_1/L_2$,
implying that the class of context-free languages
would be closed under quotient, a contradiction.
\qed
\end{proof}

\section{State Complexity of Reduced Representation}
Here we look at some bounds on the state complexity of the reduced representation of a regular language.

\begin{proposition}
\label{prop:scu}
For any $\epsilon$-NFA, $M  = (Q,\Gamma \cup \Gamma^{-1},\delta,q_0,F)$ with $n$ states, $\Gamma = \{ 1, 2, \ldots, k\}$ and $\Gamma^{-1} = \{ 1^{-1}, 2^{-1}, \ldots, k^{-1}\}$ for some positive integer $k$, there exists a DFA of at most $2^n(2k+2)$ states that accepts $r(L(M))$.
\end{proposition}

\begin{proof}
The upper bound follows from the algorithm used to prove
Proposition~\ref{prop:regalg}. 
The first part of the construction (i.e., the construction of the
automaton $M'$ accepting the $\vdash$-closure of the language
accepted by $M$) does not increase the number of states.
The resulting automaton $M'$ can be converted into a DFA with
$2^n$ states. Finally, to get an automaton accepting $r(L(M))$
we apply the usual cross-product construction to this automaton and
to the DFA with $2k+2$ states accepting $r(\Sigma^*)$ obtained
in Lemma~\ref{lem:sigstar}. The intersection results in a DFA of no more than $2^n(2k+2)$ states.
\qed
\end{proof}

Since each DFA is \emph{a fortiori} an $\epsilon$-NFA, the previous
proposition gives an upper bound for the state complexity of
the reduced representation.

\section{Length of Shortest String Reducing to the Empty Word}

Another interesting question is: given a DFA $M$ of $n$ states such
that $\epsilon \in r(L(M))$, what is the shortest $w \in L(M)$ such
that $r(w) = \epsilon$? We provide upper and lower bounds, and we
examine the special case of small alphabets. First we provide an upper
bound.

\begin{proposition}
\label{prop:lssu}
For any NFA $M= (Q,\Sigma,\delta,q_0,F)$ with $n$ states such that there exists $w \in L(M)$ with $r(w) = \epsilon$, there exists $w^\prime \in L(M)$ such that $|w^\prime| \leq 2^{n^2-n}$ and $r(w') = \epsilon$.
\end{proposition}

\begin{proof}
Suppose $M$ accepts $w \in \Sigma^+$ such that $r(w) = \epsilon$. Then $w$ can be decomposed in at least one of two ways. Either there exist $u,v \in \Sigma^+$ such that $w = uv, r(u) = \epsilon$ and $r(v) = \epsilon$ (Case 1), or there exist $u \in \Sigma^*, a \in \Sigma$ such that $w = au a^{-1}$ and $r(u) = \epsilon$ (Case 2). Any factor $w^\prime$ of $w$ such that $r(w^\prime) = \epsilon$ can also be decomposed in at least one of these two ways, so we can recursively decompose $w$ and the resulting factors until we have decomposed $w$ into single symbols. So, we can specify a certain type of parse tree such that $M$ accepts $w \in \Sigma^*$ with $r(w) = \epsilon$ if and only if we can build this type of parse tree for $w$. 

Define our parse tree for a given $w$ as follows. Every internal node corresponds to a factor $w^\prime$ of $w$ such that $r(w^\prime) = \epsilon$, and the root of the whole tree corresponds to $w$. The leaves store individual symbols. When read from left to right, the symbols in the leaves of any subtree form the word that corresponds to the root of the subtree. Each internal node is of one of two types:
\begin{enumerate}
\item The node has two children, both of which are internal nodes that serve as roots of subtrees (corresponds to Case 1).
\item The node has three children, where the left and the right children are single symbols that are inverses of each other, and the child in the middle is empty or it is an internal node that is the root of another subtree (corresponds to Case 2).
\end{enumerate}

An example is shown in Figure \ref{fig:lssu}. 
Now, we fix an accepting computation of $M$ on input $w$. 
We label each internal node $t$ with a pair of states $p,q\in Q$ 
such that if
$w'$ is the factor of $w$ that corresponds to the subtree rooted
at $t$, and $w=xw'y$, then $p\in\delta(q_0, x)$ and $q\in\delta(p, w')$
are the states reached after reading the input prefixes $x$ and $xw'$,
respectively, during the accepting computation under consideration.
(This also implies that $q_f\in\delta(q,y)$, with $q_f\in F$, and
$(q_0,q_f)$ is the label associated with the root of the tree.)

\begin{figure}[H]
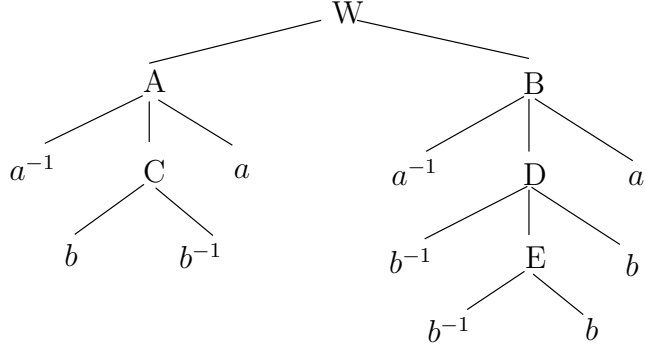

\begin{center}
\input parsetree.pstex_t
\end{center}
\caption{
	An example parse tree for the word $w = a^{-1}bb^{-1}aa^{-1}b^{-1}b^{-1}bba$, without the state pair labels. 
}
\label{fig:lssu}
\end{figure}

If the parse tree of $w$ has two nodes $t$ and $u$ with the same state-pair
label such that $u$ is a descendent of $t$, then there exists a word
shorter than $w$ which is accepted and reduces to the empty word. This is
because we can replace the subtree rooted at $t$ with the subtree
rooted at $u$. 
Furthermore, if an internal node $t$ is labeled with a pair $(q,q)$,
for some $q\in Q$, then the factor $w'$ corresponding to the subtree
rooted at $t$ can be removed from $w$, obtaining a shorter reducible
word.
Hence, by a pigeonhole argument, we conclude that the height of the
subtree corresponding to the shortest reducible word $w$ is at most
$n^2-n$. 
We now observe that the number of leaves of a parse tree
of height $k$ defined according to our rules is at most $2^k$.
(Such a tree is given by the complete binary tree of
height $k$, which has no nodes with three children. The avoidance of
nodes with three children is important because such nodes fail to
maximize the number of internal nodes in the tree, which in turn
results in less than the maximum number of leaves.)
This permits us to conclude that $|w|\leq 2^{n^2-n}$.
\qed
\end{proof}

Now we show that there is a lower bound that is exponential in the
alphabet size and in the number of the states.

\begin{proposition}
\label{prop:lss1}
For all integers $n \geq 3$ there exists a DFA, $M_n$, with $n + 1$ states over the alphabet $\Sigma = \Gamma \cup \Gamma^{-1}$, where $\Gamma = \{ 1, 2, \ldots, n-2\}$ and $\Gamma^{-1} = \{ 1^{-1}, 2^{-1}, \ldots, (n-2)^{-1}\}$, 
with the property that if $w \in L(M_n)$ and $r(w) = \epsilon$,
then $|w| \geq 2^{n-1}$. 
\end{proposition}

\begin{proof}
The proof is constructive. 
Let $M_n$ be the DFA $(Q, \Sigma, \delta, q_0, F)$, illustrated in 
Figure~\ref{fig:lss1},
where $Q = \{q_{-1}, q_0, q_1,\ldots, q_{n-1}\}$, $F = \{q_1\}$,
and
\begin{equation}
\label{eq:deltalss1}
\delta(q_a, c) =
\begin{cases}
q_1, & \text{if }c = 1 \text{ and } a = 0; \\
q_{a+1}, & \text{if }c = a^{-1} \text{ and } 1 \leq a \leq n - 2;\\
q_0, &
\begin{array}[t]{rl}
\text{if either} &c = a \text{ and } 1 \leq a \leq n - 2,\\
		\text{or} &c = 1^{-1} \text{ and } a = n-1.
\end{array}
\end{cases} \nonumber
\end{equation}
Any other transitions lead to the dead state $q_{-1}$.

\begin{figure}[H]
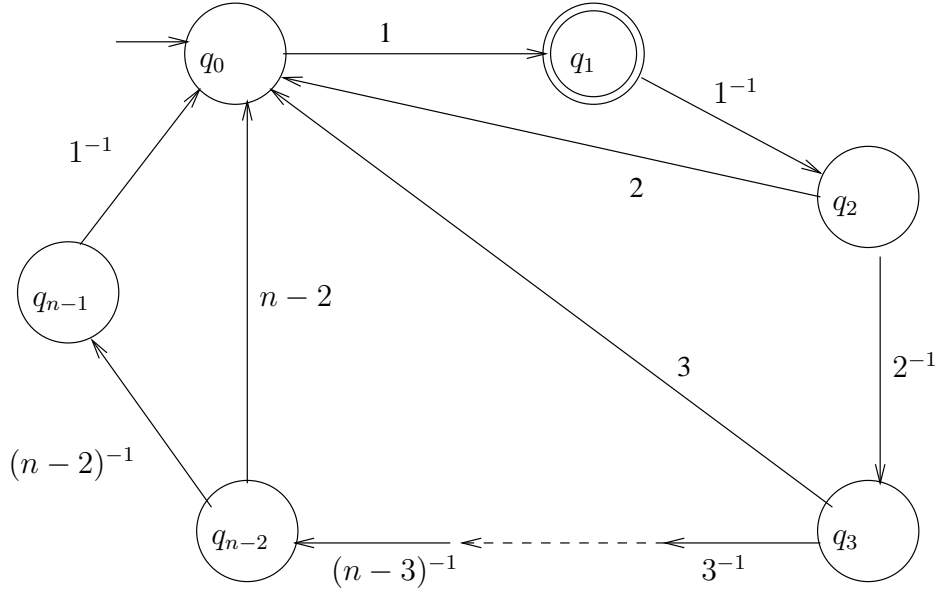

\begin{center}
\input reduce.pstex_t
\end{center}

\caption{
	$M_n$: an $n+1$ state DFA with the property that for all $w \in L(M_n)$ such that $r(w) = \epsilon$, $|w| \geq 2^{n-1}$.   The dead state is not
	shown.
}
\label{fig:lss1}
\end{figure}

Now we show that $M_n$ has the desired property. Assume there exists $w
\in L(M_n)$ such that $r(w) = \epsilon$. Then $|w|_a = |w|_{a^{-1}}$
for all $a \in \Sigma$. Since all words in $L(M_n)$ must contain the
symbol $1$ (due to the single incoming transition to the only accepting
state), it follows that $w$ must also contain $1^{-1}$.
Furthermore, the only possible transition from $q_1$ not leading to
the dead state uses the symbol $1^{-1}$.
Hence $w$
must begin with the prefix $11^{-1}$. Since $\delta(q_0, 11^{-1}) =
q_2$ and the only two transitions that leave $q_2$ are on $2$ and
$2^{-1}$, $w$ must contain both of $2$ and $2^{-1}$. Now assume that
$w$ contains the symbol $i^{-1}$ with $1 < i < n-2$. Then the state
$q_{i+1}$ must be reached while reading $w$, thus implying that the 
symbols $(i+1)$ and $(i+1)^{-1}$ also appear in $w$. 
Therefore, by induction, $w$ must contain at
least one occurrence of each $a\in \Sigma$.

Now we claim that $w$ must contain at least $2^{n-2-a}$ occurrences of
the symbol $a$, for $1 \leq a \leq n-2$, and hence at least $2^{n-2-a}$
occurrences of the symbol $a^{-1}$. 

We prove the claim by induction on $n-2-a$.
The basis, $a= n-2$, follows from the previous argument.  
Now, suppose the claim true for $k=n-2-a$. We prove it for
$k+1=n-2-(a-1)$.
By the induction hypothesis, $w$ contains at least $2^k$ occurrences
of the symbol $a$ and at least $2^k$ occurences of $a^{-1}$.
Observing the structure of the automaton, we conclude that to
have such a number of occurrences of the two letters $a$ and
$a^{-1}$, the state $q_k$ must be visited at least $2^{k+1}$ times.
On the other hand, the only transition entering $q_k$ is from the
state $q_{k-1}$ on the letter $(a-1)^{-1}$. Hence, $w$ must contain
at least $2^{k+1}=2^{n-2-(a-1)}$ occurrences of $(a-1)^{-1}$ and,
according to the initial discussion, at least 
$2^{k+1}=2^{n-2-(a-1)}$ occurrences of $a-1$. This proves the claim.

By computing the sum over all alphabet symbols, we get that
$|w| \geq 2 (2^{n-2}-1)$. However, since
the symbol $(n-2)^{-1}$ must always be followed by the symbol $1^{-1}$,
$w$ must actually contain one additional occurrence of each of $1$ and
$1^{-1}$. Thus $|w| \geq 2^{n-1}$.
\qed
\end{proof}

It turns out that the shortest word that reduces to $\epsilon$ accepted
by the DFA $M_n$ in the proof of Proposition \ref{prop:lss1} is related
to the well-known ruler sequence, $(\nu_2 (n))_{n \geq 1}$, where
$\nu_2 (n) $ denotes the exponent of the highest power of $2$ dividing
$n$.
This sequence has many
interesting characterizations including being the lexicographically
least infinite squarefree word over $\Zee$.
For integers $k > 0$, let $r_k = (\nu_2(n)_{1 \leq n <
2^k})$ be the prefix of length $2^k - 1$ of the ruler sequence.  For example,
$r_3 = 0102010$.

Let $w_n$ be the shortest word accepted by $M_n$
that reduces to $\epsilon$.
Let $w^\prime_n$ be the prefix of $w$ of length $|w|-2$.
This is well defined for $n\geq 3$. Also let $w_n = w^\prime_n =
\epsilon$ for $n=2$. Then for any integer $n \geq 3$, we have $w_n =
w^\prime_{n-1}(n-2)w^\prime_{n-1}(n-2)^{-1}1^{-1}1$. It can be easily verified
that this word is accepted by $M_n$.
Now define
the homomorphism $h$ such that $h(a) = a$ for $a\in \Gamma$, and $h(a)
= \epsilon$ for $a\in\Gamma^{-1}$. Then $h(w_n) = r_{n-2}1$.

The next proposition shows that over a fixed alphabet size we can still get an exponential lower bound.

\begin{proposition}
\label{prop:lss2}
For each integer $n \geq 1$ there exists a DFA, $M_n$ with $3n + 1$ states over the alphabet $\Sigma = \Gamma \cup \Gamma^{-1}$, where $\Gamma = \{ 1, 2\}$ and $\Gamma^{-1} = \{ 1^{-1}, 2^{-1}\}$, with the property that the only word $w \in L(M_n)$ such that $r(w) = \epsilon$ has length $|w| = 3\cdot 2^{n} - 4$. 
\end{proposition}

\begin{proof}
The proof is constructive. 
Let $M_n$ be the DFA $(Q, \Sigma, \delta, q_n, F)$ where $Q = \{q_{-1}, q_0, q_1, p_1\} \cup\{p_i,q_i,r_i: 2\leq i \leq n\}$, $F = \lbrace p_n \rbrace$, and
\begin{equation}
\label{eq:deltalss2q}
\delta(q_a, c) = q_{a-1}, 
\begin{array}[t]{rl}
\text{  if }1 \leq a \leq n,\text{ and either}
						&c = 1\text{ and }a \equiv 1\!\!\!\!\pmod 2,\\
						\text{or}&c = 2\text{ and }a \equiv 0\!\!\!\! \pmod 2.
\end{array}
\nonumber
\end{equation}

\begin{equation}
\label{eq:deltalss2p}
\delta(p_a, c) =
\begin{cases}
p_{a+1}, 
&\begin{array}[t]{rl}
\text{if }1 \leq a \leq n-1,\text{ and either}
						&c = 1\text{ and }a \equiv 0\!\!\!\! \pmod 2,\\
						\text{or}&c = 2\text{ and }a \equiv 1\!\!\!\! \pmod 2,
\end{array}\\
r_{a+1}, 
&\begin{array}[t]{rl}
\text{if }1 \leq a \leq n-1,\text{ and either}
						&c = 1^{-1}\text{ and }a \equiv 0\!\!\!\! \pmod 2,\\
						\text{or}&c = 2^{-1}\text{ and }a \equiv 1\!\!\!\! \pmod 2.
\end{array}
\end{cases}
\nonumber
\end{equation}

\begin{equation}
\label{eq:deltalss2r}
\delta(r_a, c) = q_{a-1},
\begin{array}[t]{rl}
\text{  if }2 \leq a \leq n,\text{ and either}
						&c = 1^{-1}\text{ and }a \equiv 1\!\!\!\!\pmod 2,\\
						\text{or}&c=2^{-1}\text{ and }a \equiv 0\!\!\!\! \pmod 2.
\end{array}
\nonumber
\end{equation}

\begin{equation}
\label{eq:deltalss2q0}
\delta(q_0,a^{-1})=p_1
\nonumber
\end{equation}

Any other transitions lead to the dead state $q_{-1}$.
The case $n=4$ is illustrated in Figure \ref{fig:lss2}.

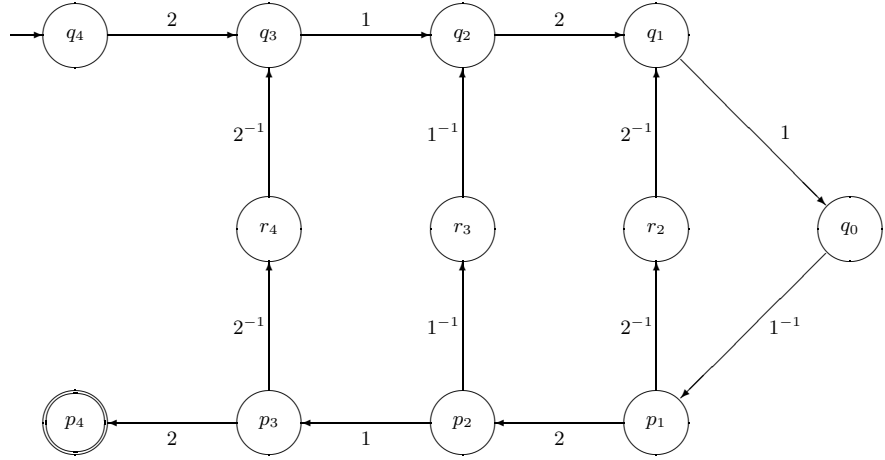
\begin{figure}[H]
\begin{center}

\setlength{\unitlength}{0.5cm}
 \resizebox{\columnwidth}{!}
{
  \begin{picture}(28,16)(0,0)

      \put(3,2){\oval(2,2)}
      \put(3,2){\makebox(0,0){$p_4$}}
      \put(9,2){\oval(2,2)}
      \put(9,2){\makebox(0,0){$p_3$}}
      \put(15,2){\oval(2,2)}
      \put(15,2){\makebox(0,0){$p_2$}}
      \put(21,2){\oval(2,2)}
      \put(21,2){\makebox(0,0){$p_1$}}

      \put(3,14){\oval(2,2)}
      \put(3,14){\makebox(0,0){$q_4$}}
      \put(9,14){\oval(2,2)}
      \put(9,14){\makebox(0,0){$q_3$}}
      \put(15,14){\oval(2,2)}
      \put(15,14){\makebox(0,0){$q_2$}}
      \put(21,14){\oval(2,2)}
      \put(21,14){\makebox(0,0){$q_1$}}

      \put(9,8){\oval(2,2)}
      \put(9,8){\makebox(0,0){$r_4$}}
      \put(15,8){\oval(2,2)}
      \put(15,8){\makebox(0,0){$r_3$}}
      \put(21,8){\oval(2,2)}
      \put(21,8){\makebox(0,0){$r_2$}}
      \put(27,8){\oval(2,2)}
      \put(27,8){\makebox(0,0){$q_0$}}
      
      \put(1,14){\vector(1,0){1}}
      \put(3,2){\oval(1.8,1.8)}
      
      \put(8,2){\vector(-1,0){4}}
      \put(6,1.5){\makebox(0,0){$2$}}
      \put(14,2){\vector(-1,0){4}}
      \put(12,1.5){\makebox(0,0){$1$}}
      \put(20,2){\vector(-1,0){4}}
      \put(18,1.5){\makebox(0,0){$2$}}
      
      \put(4,14){\vector(1,0){4}}
      \put(6,14.5){\makebox(0,0){$2$}}
      \put(10,14){\vector(1,0){4}}
      \put(12,14.5){\makebox(0,0){$1$}}
      \put(16,14){\vector(1,0){4}}
      \put(18,14.5){\makebox(0,0){$2$}}
      
      \put(9,3){\vector(0,1){4}}
      \put(8.4,5){\makebox(0,0){$2^{-1}$}}
      \put(9,9){\vector(0,1){4}}
      \put(8.4,11){\makebox(0,0){$2^{-1}$}}
      \put(15,3){\vector(0,1){4}}
      \put(14.4,5){\makebox(0,0){$1^{-1}$}}
      \put(15,9){\vector(0,1){4}}
      \put(14.4,11){\makebox(0,0){$1^{-1}$}}
      \put(21,3){\vector(0,1){4}}
      \put(20.4,5){\makebox(0,0){$2^{-1}$}}
      \put(21,9){\vector(0,1){4}}
      \put(20.4,11){\makebox(0,0){$2^{-1}$}}
      
      \put(21.75,13.25){\vector(1,-1){4.5}}
      \put(25,11){\makebox(0,0){$1$}}
      \put(26.25,7.25){\vector(-1,-1){4.5}}
      \put(25,5){\makebox(0,0){$1^{-1}$}}

   \end{picture}
}
   \caption{$M_4$: a $3\cdot 4+1$ state DFA with the property that 
   the only reducible word accepted by it has length $3 \cdot 2^{4} - 4$.
   (The dead state is not represented.)}
   \label{fig:lss2}

\end{center}
\end{figure}

In order to prove the statement, for each integer $m\geq 0$, let us consider
the set $C_m$ of pairs of states, others than the dead state,
which are connected by a reducible word of length $m$, i.e.
\[
C_m=\{(s',s'')\in Q'\times Q':\exists w\in\Sigma^*\text{ s.t.\ }
|w|=m, r(w)=\epsilon,\text{ and }\delta(s',w)=s''\},
\]
where $Q'=Q\setminus\{q_{-1}\}$.
We notice that $C_m=\emptyset$, for $m$ odd.
Furthermore $C_0=\{(s,s):s\in Q'\}$ and, for $m>0$,
$C_m=C'_m\cup C''_m$, where:
\begin{eqnarray}
C'_m = \{(s',s'')&:&\exists(r',r'')\in C_{m-2}, a\in\Sigma\nonumber\\
&&\text{ s.t.\ }\delta(s',a)=r'\text{ and }\delta(r'',a^{-1})=s''\},
\label{eq:cprime}
\end{eqnarray}
\begin{eqnarray}
C''_m = \{(s',s'')&:&\exists m',m''>0, (s',r')\in C_{m'},
	(r'',s'')\in C_{m''}\nonumber\\
	&&\text{ s.t.\ }m'+m''=m\text{ and }r'=r''\}.\label{eq:csecond}
\end{eqnarray}
We claim that, for each $m\geq 1$:
\begin{equation}
C_m=\left\{
\begin{array}{l@{~~}l}
\{(q_k,p_k)\}, 			& \mbox{if $\exists k$, $1\leq k\leq n$, s.t.\ $m=3\cdot 2^k-4$;}\\
\{(q_k,r_k),(r_k,p_k)\},	& \mbox{if $\exists k$, $2\leq k\leq n$, s.t.\ $m=3\cdot 2^{k-1}-2$;}\\
\emptyset,				& \mbox{otherwise.}
\end{array}
\right.\label{eq:cm}
\end{equation}
We prove (\ref{eq:cm}) by induction on $m$.

As already noticed, $C_1=\emptyset$. By inspecting the transition
function of $M_n$, we can observe that
$C_2=\{(q_1,p_1)\}$. Notice that $2=3\cdot 2^1-4$. This proves the
basis.

For the inductive step, we now consider $m>2$ and we suppose
(\ref{eq:cm}) true for integers less than $m$.

First, we show that we can simplify the formula (\ref{eq:csecond})
for $C''_m$.
In fact, using the inductive hypothesis, 
for $0<m',m''<m$, the only possible $(s',r')\in C_{m'}$
and $(r'',s'')\in C_{m''}$ satisfying $r'=r''$ are the pairs 
$(q_j,r_j),(r_j,p_j)\in C_{3\cdot 2^{j-1}-2}$, obtained by taking $m'=m''=3\cdot 2^{j-1}-2$,
for suitable values of $j$. This, together with the condition $m'+m''=m$, restricts the set $C''_m$ to:
\begin{equation}
C''_m = \{(s',s'')\mid\exists r\in Q: (s',r)\in C_{m/2}\text{ and }(r,s'')\in C_{m/2}\}.
\label{eq:csecshort}
\end{equation}

Now we consider three subcases:

\emph{Case~1: $m=3\cdot 2^k-4$}, with $k\geq 2$.\\
An easy verification shows that $m-2$ cannot be expressed
in the form $3\cdot 2^j-4$ or in the form $3\cdot 2^j-2$, for any $j$.
Hence, by the inductive hypothesis, $C_{m-2}=\emptyset$.
By (\ref{eq:cprime}), this implies $C'_m=\emptyset$, and then $C_m=C''_m$.

We now compute $C''_m$ using (\ref{eq:csecshort}) and the set
$C_{m/2}$ obtained according to the inductive hypothesis.
We observe that $m/2=3\cdot 2^{k-1}-2$. 
Hence, for $k\leq n$, $C_{m/2}=\{(q_k,r_k),(r_k,p_k)\}$ and, thus,
$C_m=C''_m=\{(q_k,p_k)\}$. On the other hand, if $k>n$ then 
$C_{m/2}=\emptyset$, which implies $C_m=C''_m=\emptyset$.

\smallskip
\emph{Case~2: $m=3\cdot 2^{k-1}-2$}, with $k\geq 2$.\\
First, we observe that $m/2$ cannot be written either as $3\cdot 2^j-4$ 
or as $3\cdot 2^j-2$. Hence, by the inductive hypothesis, the set $C''_m$ 
must be empty. Thus, $C_m=C'_m$.

We compute $C'_m$ as in (\ref{eq:cprime}), using the set $C_{m-2}$ obtained
according to the induction hypothesis.
We notice that $m-2=3\cdot 2^{k-1}-4$.
If $k>n+1$ then $C_{m-2}=\emptyset$, thus implying $C_m=C'_m=\emptyset$.
Otherwise, $C_{m-2}=\{(q_{k-1},p_{k-1})\}$. 
In order to obtain all the elements of $C'_m$, we have to examine
the transitions entering $q_{k-1}$ or leaving $p_{k-1}$.
For $k=n+1$ there are no such transitions and, hence, $C_m=C'_m=\emptyset$.
For $k\leq n$ all the transitions entering $q_{k-1}$ or leaving $p_{k-1}$ 
involve the same symbol $a\in\{1,2\}$ or its inverse: there are exactly 
two transitions entering $q_{k-1}$ ($\delta(q_k,a)=\delta(r_k,a^{-1})=q_{k-1}$) 
and exactly two transitions leaving $p_{k-1}$ ($\delta(p_{k-1},a)=p_k$ and
$\delta(p_{k-1},a^{-1})=r_k$). Hence, by the appropriate combinations 
of these transitions with the only pair $(q_{k-1},p_{k-1})$ in 
$C_{m-2}$, we get that $C_m=C'_m=\{(q_k,r_k),(r_k,p_k)\}$.

\smallskip
\emph{Case~3: Remaining values of $m$}.\\
If $m=3\cdot 2^{k-1}$ with $k\geq 2$, then $C_{m-2}=\{(q_k,r_k),(r_k,p_k)\}$.
All the transitions entering or leaving $r_k$ use the same symbol $a^{-1}$, with
$a\in\{1,2\}$, while all the transitions entering $q_k$ or leaving $p_k$ use
the other symbol $b\in\{1,2\}$, $b\neq a$, or $b^{-1}$. Hence, 
from (\ref{eq:cprime}), $C'_m=\emptyset$.

For all the other values of $m$, the form of $m-2$ is neither 
$3\cdot 2^j-4$ nor $3\cdot 2^j-2$. This implies that
$C_{m-2}=\emptyset$ and, then, $C'_m$ must be empty.
Hence, we conclude that $C_m=C''_m$.

Suppose $C''_m\neq\emptyset$. From
(\ref{eq:csecshort}) and the inductive hypothesis, it follows that 
$m/2=3\cdot 2^{j-1}-2$ for some $j$, thus implying $m=3\cdot 2^j-4$.
This is a contradiction, because the values of $m$ we are considering
are not of this form. Hence, $C_m=C''_m$ must be empty.

This completes the proof of (\ref{eq:cm}).

\medskip
Recall that the initial state of $M_n$ is $q_n$, while the only final
state is $p_n$. Hence, the length of the shortest reducible word
accepted by $M_n$ is the smallest integer $m$ such that
$(q_n,p_n)\in C_m$.
According to (\ref{eq:cm}), we conclude that such a length
is $3\cdot 2^n-4$. 

From (\ref{eq:cm}), it also follows that there are no reducible words
accepted by $M_n$ whose length is different than $3\cdot 2^n-4$.
With some small refinements in the argument used 
to prove (\ref{eq:cm}),
we can show that $M_n$ accepts exactly one reducible word.
In particular, for $k\geq 1$ we consider:
\[
w_k=\left\{
\begin{array}{l@{~~}l}
11^{-1}, 					& \mbox{if $k=1$;}\\
1w_{k-1}1^{-1}1^{-1}w_{k-1}1,	& \mbox{if $k>1$ and $k$ odd;}\\
2w_{k-1}2^{-1}2^{-1}w_{k-1}2,	& \mbox{otherwise.}
\end{array}
\right.
\]
By an inductive argument it can be proved that $w_k$ is the only
reducible word such that $\delta(q_k,w_k)=p_k$ and $|w_k|=3\cdot 2^k-4$.
\qed
\end{proof}

We now examine the special case where $\Sigma = \{1,1^{-1}\}$, and give a cubic upper bound and quadratic lower bound. The next proposition gives the upper bound, which
holds even in the nondeterministic case.

\begin{proposition}
\label{prop:lssonegenup}
Let $M  = (Q,\{1,1^{-1}\},\delta,q_0,F)$ be an NFA with $n$ states such
that $\epsilon\in r(L(M))$. Then $M$ accepts a reducible word of
length at most $n(2n^2+1)$.
\end{proposition}

\begin{proof}
We prove the result by contradiction. Assume the shortest $w \in L(M)$
such that $r(w) = \epsilon$ has $|w| > n(2n^2+1)$. Define a function
$b$ on words by $b(z) = |z|_1 - |z|_{1^{-1}}$ for $z \in \Sigma^*$.
Roughly speaking, the function $b$ measures the ``balance'' between the
number of occurrences of the symbol $1$ and those of the symbol
$1^{-1}$ in a word. 

Suppose that no factor $w'$ of $w$ has $|b(w')| > n^2$. 
Then the function $b$ can take on at most $2n^2+1$
distinct values. Since $|w| > n(2n^2+1)$, there must be a value $C$
such that $b$ takes the value $C$ for more than $n$ different prefixes
of $w$.  That is, there is some $\ell\geq n$ such that $w = x y_1 y_2
\cdots y_{\ell} z$ where the $y_i$ are nonempty and
\[
b(x) = b(xy_1) = b(xy_1y_2) = \cdots = b(x y_1 y_2 \cdots y_{\ell}).
\]
Consider a sequence of $\ell+1$ states $p_0,p_1,\ldots,p_{\ell}\in Q$,
and a state $q_f\in F$, such that
during an accepting computation $M$ makes the following transitions:
\[p_0\in\delta(q_0,x), p_1\in\delta(p_0,y_1),
\ldots,
p_{\ell}\in\delta(p_{\ell-1},y_{\ell}), q_f\in\delta(p_{\ell},z)
\]
Since $\ell\geq n$, a state must repeat in the above sequence, say
$p_i=p_j$ with $i < j$.  Then $u = x y_1 \cdots y_i y_{j+1} \cdots y_{\ell} z$ is shorter
than $w$ (since we have omitted $y_{i+1} \cdots y_j$) and it is
accepted by $M$. Furthermore, observing that $b(y_{i+1} \cdots y_j)=0$, we
conclude that $r(u) = \epsilon$. Since $|u|<|w|$, this is a contradiction
to our choice of $w$. Hence, $w$ must contain 
a factor $w^\prime$ of $w$ such that $|b(w^\prime)|> n^2$.

Let $y$ be the shortest factor of $w$ such that $b(y) = 0$ and 
$|b(y^\prime)|> n^2$ for some prefix $y^\prime$ of $y$. 
We can write $w = xyz$, for suitable words $x,z$.
Let $D$ be the maximum value of $|b(y^\prime)|$ over all prefixes $y^\prime$
of $y$.  We suppose that $D>0$. (The argument can be easily adapted to the
case $D<0$.) For $i = 0,1,2,\ldots,D$, let $R(i)$ be the shortest prefix
of $y$ with $b(R(i)) = i$.  Similarly, let $S(i)$ be the
longest prefix of $y$ with $b(S(i)) = i$.  
Again, consider an accepting computation of $M$ on 
input $w$. For each pair
$[R(i),S(i)]$, let $[P(i),Q(i)]$ be the pair of states such that $M$ is
in state $P(i)$ after reading $xR(i)$ and $M$ is in state $Q(i)$ after
reading $xS(i)$.  Since $D > n^2$, some pair of states repeats in the
sequence $\{[P(i),Q(i)]\}$.  That is, there exists $j<k$ such that
$[P(j),Q(j)] = [P(k),Q(k)]$.  We may therefore omit the portion of the
computation that occurs between the end of $R(j)$ and the end of
$R(k)$ as well as that which occurs between the end of $S(k)$ and the
end of $S(j)$ to obtain a computation accepting a shorter word $u$ 
such that $r(u) = \epsilon$. Again we have a contradiction, and our result follows.
\qed
\end{proof}

The following proposition gives a quadratic lower bound.

\begin{proposition}
\label{prop:lssonegenlow}
For each integer $n \geq 0$ there exists a DFA $M_n$ with $n + 1$
states, over the alphabet $\Sigma = \{1,1^{-1}\}$, such that the
only reducible word $w \in L(M_n)$, has length
$(n+1)(n-1)/2 $ if $n$ is odd, and $n^2/2$
if $n$ is even.
\end{proposition}

\begin{proof}
The proof is constructive. Let $M_n$ be the DFA $(Q, \Sigma, \delta, q_0, F)$, 
illustrated in Figure \ref{fig:lssonegenlow}, 
where $Q = \{q_{-1}\}\cup\{q_i: 0\leq i <n\}$, 
$F = \{q_{\lfloor \frac{n}{2} \rfloor}\}$, and
\begin{equation}
\label{eq:deltalssonegenlow}
\delta(q_a, c) =
\begin{cases}
q_{a+1}, & \text{if either }c = 1 \text{ and } 0 \leq a < \lfloor \frac{n}{2} \rfloor,\\
~ & \text{or }c = 1^{-1} \text{ and } \lfloor \frac{n}{2} \rfloor \leq a \leq n - 2;\\
q_{a \bmod 2}, & \text{if }c = 1^{-1} \text{ and } a = n-1.\\
\end{cases} \nonumber
\end{equation}
Any other transitions lead to the dead state, $q_{-1}$. 

Observe that if $n$ is odd, then each word $w$ accepted by $M_n$ 
has the form $w = (1^{\frac{n-1}{2}}1^{-\frac{n+1}{2}})^{\alpha}1^{\frac{n-1}{2}}$,
for an $\alpha\geq 0$. Computing the ``balance'' function $b$ introduced
in the proof of Proposition~\ref{prop:lssonegenup}, we get
$b(w)=\frac12(n-1-2\alpha)$, which is $0$ if and only if $\alpha = \frac{n-1}2$.
Finally, by computing the length of $w$ for such an $\alpha$, we obtain
$(n+1)(n-1)/2$.

Similarly, in the case of $n$ even, we can prove that the only reducible word
accepted by $M_n$ has length $n^2/2$.
\qed
\end{proof}

\begin{figure}[H]
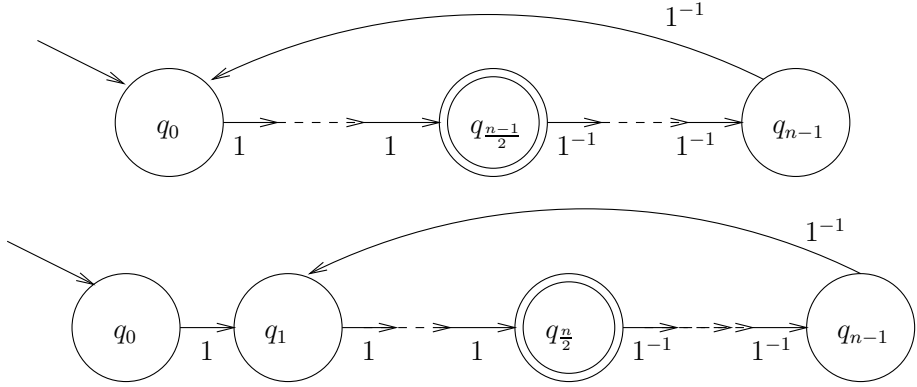

\begin{center}
\resizebox{\columnwidth}{!}{\input cycle.pstex_t}
\end{center}

\caption{
	Top: $M_n$ where $n$ is odd. Bottom: $M_n$ where $n$ is even.
}
\label{fig:lssonegenlow}
\end{figure}

\section{Counterexamples}

Here we look at counterexamples that solve two natural questions. The
first question is whether the set of all equivalent words to those in a
regular language must be regular. The second question is whether any of
our results hold for a more generalized form of the reduced
representation. We now define the set of equivalent words.

\begin{definition}
For $w\in \Sigma^*$, the set of equivalent words is $${\rm eq}(w) = \{w^\prime: r(w) = r(w^\prime)\}$$ For $L\subseteq \Sigma^*$, the set of equivalent words is 
$${\rm eq}(L) = \{{\rm eq}(w): w \in L\}$$
\end{definition}

\begin{proposition}
\label{prop:equivalent}
There exists $L \subseteq \Sigma^*$ such that $L$ is regular but ${\rm eq}(L)$ is not.
\end{proposition}

\begin{proof}
Consider the following example. Let $\Sigma = \{1,1^{-1}\}$ and let $L = \{\epsilon\}$. Then ${\rm eq}(L)$ is the language of balanced parentheses, which is 
well known to be nonregular.
\qed
\end{proof}

\begin{proposition}
\label{prop:equiv_cf}
Let $L$ be a regular language over an alphabet $\Sigma$.  The language
${\rm eq}(L)$ is context-free.
\end{proposition}

\begin{proof}
Let $M$ be a DFA accepting $L$.  We first construct the $\epsilon$-NFA
$M_r$ of Proposition~\ref{prop:regalg} that accepts $r(L)$.  We then reverse the
transitions of $M_r$ to obtain an $\epsilon$-NFA $A$ that accepts the
reversal of $r(L)$.  We now construct a PDA $B$ that accepts ${\rm
eq}(L)$.  The operation of $B$ is as follows.  On input $w$, the PDA
$B$ reads each symbol of $w$ and compares it with the symbol on top of
the stack.  If the symbol being read is $a$ and the symbol on top of
the stack is the inverse of $a$, the machine $B$ pops the top symbol
of the stack.  Otherwise, the machine $B$ pushes $a$ on top of the
stack.  After the input $w$ is consumed, the stack contains a word $z$
that is equivalent to $w$.  Moreover, since $z$ does not contain any
factor of the form $aa^{-1}$, the word $z$ must equal $r(w)$ by
Lemma~\ref{lem:red}.  Finally, on $\epsilon$-transitions, the PDA $B$
pops each symbol of $z$ off the stack and simulates the computation of
$A$ on each popped symbol.  The net effect is to simulate $A$ on
$z^R$ (the reversal of $z$).  If $z^R$ is accepted by $A$, the PDA $B$
accepts $w$.  Since $z^R$ is accepted by $A$ if and only if $z \in
r(L)$, the PDA $B$ accepts $w$ if and only if $r(w) \in r(L)$.
However, we have $r(w) \in r(L)$ if and only if $w \in {\rm eq}(L)$,
so $B$ accepts ${\rm eq}(L)$, as required.
\qed
\end{proof}

A set of equivalent words as described above can be thought of as an equivalence class under an equivalence relation described by a very particular set of equations: for all $a\in \Sigma, aa^{-1} = \epsilon$. Our generalization is to allow for an arbitrary set of equations, which we will refer to as the ``defining set of equations''. Then a generalized reduced representation of a word is an equivalent word such that there are no shorter equivalent words. It is no longer necessary that reduced representation be unique, so we denote the set of generalized reduced representations of a word $w$ as $r_g(w)$. For example, if $\Sigma = \{a, b, c, d\}$ we may have the defining set of equations $\{ab = cd, bc = a\}$. Then ${\rm eq}(abd) = \{abd, cdd, bcbd\}$ and $r_g(aabd) = \{abd, cdd\}$. It is natural to wonder whether an analogous result to Corollary \ref{cor:regclosure} holds under this generalized form of reduced representations.

\begin{proposition}
\label{prop:general}
There exists $L \subseteq \Sigma^*$ such that $L$ is regular but
$r_g(L)$ is not even context-free.
\end{proposition}

\begin{proof}
Consider the following example. Let $\Sigma = \{a,b,c\}$, let the
defining set of equations be $\{ab = ba, ac = ca, bc = cb\}$ and let $L
= (abc)^*$. Then $r_g(L)$ is the language $\{x: |x|_a = |x|_b =
|x|_c\}$, which is known to be noncontext-free, and hence also
nonregular.
\qed
\end{proof}

\end{document}

%% file: rsigstar.pstex_t
\begin{picture}(0,0)%
\epsfig{file=rsigstar.pstex}%
\end{picture}%
\setlength{\unitlength}{3947sp}%
\begingroup\makeatletter\ifx\SetFigFont\undefined%
\gdef\SetFigFont#1#2#3#4#5{%
  \reset@font\fontsize{#1}{#2pt}%
  \fontfamily{#3}\fontseries{#4}\fontshape{#5}%
  \selectfont}%
\fi\endgroup%
\begin{picture}(4331,2609)(193,-1980)
\put(2401,359){\makebox(0,0)[lb]{\smash{\SetFigFont{10}{12.0}{\rmdefault}{\mddefault}{\updefault}$q_0$}}}
\put(608,-608){\makebox(0,0)[lb]{\smash{\SetFigFont{10}{12.0}{\rmdefault}{\mddefault}{\updefault}$q_{1}$}}}
\put(2963,-623){\makebox(0,0)[lb]{\smash{\SetFigFont{10}{12.0}{\rmdefault}{\mddefault}{\updefault}$q_2$}}}
\put(1449,-331){\makebox(0,0)[lb]{\smash{\SetFigFont{10}{12.0}{\rmdefault}{\mddefault}{\updefault}1}}}
\put(2708, 45){\makebox(0,0)[lb]{\smash{\SetFigFont{10}{12.0}{\rmdefault}{\mddefault}{\updefault}2}}}
\put(3204,-960){\makebox(0,0)[lb]{\smash{\SetFigFont{10}{12.0}{\rmdefault}{\mddefault}{\updefault}$2^{-1}$}}}
\put(2378,-1741){\makebox(0,0)[lb]{\smash{\SetFigFont{10}{12.0}{\rmdefault}{\mddefault}{\updefault}$q_3$}}}
\put(3287, 37){\makebox(0,0)[lb]{\smash{\SetFigFont{10}{12.0}{\rmdefault}{\mddefault}{\updefault}$2^{-1}$}}}
\put(2087,-1374){\makebox(0,0)[lb]{\smash{\SetFigFont{10}{12.0}{\rmdefault}{\mddefault}{\updefault}1}}}
\put(2356,-788){\makebox(0,0)[lb]{\smash{\SetFigFont{10}{12.0}{\rmdefault}{\mddefault}{\updefault}2}}}
\put(1516,-960){\makebox(0,0)[lb]{\smash{\SetFigFont{10}{12.0}{\rmdefault}{\mddefault}{\updefault}2}}}
\put(280,-653){\makebox(0,0)[lb]{\smash{\SetFigFont{10}{12.0}{\rmdefault}{\mddefault}{\updefault}1}}}
\put(1135,-623){\makebox(0,0)[lb]{\smash{\SetFigFont{10}{12.0}{\rmdefault}{\mddefault}{\updefault}$1^{-1}$}}}
\put(3384,-382){\makebox(0,0)[lb]{\smash{\SetFigFont{10}{12.0}{\rmdefault}{\mddefault}{\updefault}$1^{-1}$}}}
\put(2964,-1298){\makebox(0,0)[lb]{\smash{\SetFigFont{10}{12.0}{\rmdefault}{\mddefault}{\updefault}$2^{-1}$}}}
\put(3616,-1321){\makebox(0,0)[lb]{\smash{\SetFigFont{10}{12.0}{\rmdefault}{\mddefault}{\updefault}2}}}
\put(2679,-1500){\makebox(0,0)[lb]{\smash{\SetFigFont{10}{12.0}{\rmdefault}{\mddefault}{\updefault}$2^{-1}$}}}
\put(2318,-465){\makebox(0,0)[lb]{\smash{\SetFigFont{10}{12.0}{\rmdefault}{\mddefault}{\updefault}$1^{-1}$}}}
\put(1719,-61){\makebox(0,0)[lb]{\smash{\SetFigFont{10}{12.0}{\rmdefault}{\mddefault}{\updefault}1}}}
\put(1981,105){\makebox(0,0)[lb]{\smash{\SetFigFont{10}{12.0}{\rmdefault}{\mddefault}{\updefault}$1^{-1}$}}}
\put(1630,127){\makebox(0,0)[lb]{\smash{\SetFigFont{10}{12.0}{\rmdefault}{\mddefault}{\updefault}1}}}
\put(1695,-600){\makebox(0,0)[lb]{\smash{\SetFigFont{10}{12.0}{\rmdefault}{\mddefault}{\updefault}$q_{1^{-1}}$}}}
\put(4112,-616){\makebox(0,0)[lb]{\smash{\SetFigFont{10}{12.0}{\rmdefault}{\mddefault}{\updefault}$q_{2^{-1}}$}}}
\put(3526,-646){\makebox(0,0)[lb]{\smash{\SetFigFont{10}{12.0}{\rmdefault}{\mddefault}{\updefault}$2^{-1}$}}}
\put(1502,-1508){\makebox(0,0)[lb]{\smash{\SetFigFont{10}{12.0}{\rmdefault}{\mddefault}{\updefault}$1^{-1}$}}}
\put(2618,-661){\makebox(0,0)[lb]{\smash{\SetFigFont{10}{12.0}{\rmdefault}{\mddefault}{\updefault}2}}}
\end{picture}

%% file: parsetree.pstex_t
\begin{picture}(0,0)%
\epsfig{file=parsetree.pstex}%
\end{picture}%
\setlength{\unitlength}{3947sp}%
\begingroup\makeatletter\ifx\SetFigFont\undefined%
\gdef\SetFigFont#1#2#3#4#5{%
  \reset@font\fontsize{#1}{#2pt}%
  \fontfamily{#3}\fontseries{#4}\fontshape{#5}%
  \selectfont}%
\fi\endgroup%
\begin{picture}(3919,2145)(376,-1846)
\put(1216,-271){\makebox(0,0)[lb]{\smash{\SetFigFont{12}{14.4}{\rmdefault}{\mddefault}{\updefault}A}}}
\put(3601,-286){\makebox(0,0)[lb]{\smash{\SetFigFont{12}{14.4}{\rmdefault}{\mddefault}{\updefault}B}}}
\put(376,-826){\makebox(0,0)[lb]{\smash{\SetFigFont{12}{14.4}{\rmdefault}{\mddefault}{\updefault}$a^{-1}$}}}
\put(1216,-841){\makebox(0,0)[lb]{\smash{\SetFigFont{12}{14.4}{\rmdefault}{\mddefault}{\updefault}C}}}
\put(1786,-811){\makebox(0,0)[lb]{\smash{\SetFigFont{12}{14.4}{\rmdefault}{\mddefault}{\updefault}$a$}}}
\put(2776,-856){\makebox(0,0)[lb]{\smash{\SetFigFont{12}{14.4}{\rmdefault}{\mddefault}{\updefault}$a^{-1}$}}}
\put(3601,-856){\makebox(0,0)[lb]{\smash{\SetFigFont{12}{14.4}{\rmdefault}{\mddefault}{\updefault}D}}}
\put(4261,-856){\makebox(0,0)[lb]{\smash{\SetFigFont{12}{14.4}{\rmdefault}{\mddefault}{\updefault}$a$}}}
\put(721,-1366){\makebox(0,0)[lb]{\smash{\SetFigFont{12}{14.4}{\rmdefault}{\mddefault}{\updefault}$b$}}}
\put(1441,-1381){\makebox(0,0)[lb]{\smash{\SetFigFont{12}{14.4}{\rmdefault}{\mddefault}{\updefault}$b^{-1}$}}}
\put(2761,-1411){\makebox(0,0)[lb]{\smash{\SetFigFont{12}{14.4}{\rmdefault}{\mddefault}{\updefault}$b^{-1}$}}}
\put(3616,-1381){\makebox(0,0)[lb]{\smash{\SetFigFont{12}{14.4}{\rmdefault}{\mddefault}{\updefault}E}}}
\put(4246,-1426){\makebox(0,0)[lb]{\smash{\SetFigFont{12}{14.4}{\rmdefault}{\mddefault}{\updefault}$b$}}}
\put(3001,-1846){\makebox(0,0)[lb]{\smash{\SetFigFont{12}{14.4}{\rmdefault}{\mddefault}{\updefault}$b^{-1}$}}}
\put(3991,-1831){\makebox(0,0)[lb]{\smash{\SetFigFont{12}{14.4}{\rmdefault}{\mddefault}{\updefault}$b$}}}
\put(2401,164){\makebox(0,0)[lb]{\smash{\SetFigFont{12}{14.4}{\rmdefault}{\mddefault}{\updefault}W}}}
\end{picture}

%% file: reduce.pstex_t
\begin{picture}(0,0)%
\epsfig{file=reduce.pstex}%
\end{picture}%
\setlength{\unitlength}{3947sp}%
\begingroup\makeatletter\ifx\SetFigFont\undefined%
\gdef\SetFigFont#1#2#3#4#5{%
  \reset@font\fontsize{#1}{#2pt}%
  \fontfamily{#3}\fontseries{#4}\fontshape{#5}%
  \selectfont}%
\fi\endgroup%
\begin{picture}(5726,3650)(301,-3236)
\put(1501, 14){\makebox(0,0)[lb]{\smash{\SetFigFont{12}{14.4}{\familydefault}{\mddefault}{\updefault}$q_0$}}}
\put(3826, 14){\makebox(0,0)[lb]{\smash{\SetFigFont{12}{14.4}{\familydefault}{\mddefault}{\updefault}$q_1$}}}
\put(4726,-211){\makebox(0,0)[lb]{\smash{\SetFigFont{12}{14.4}{\familydefault}{\mddefault}{\updefault}$1^{-1}$}}}
\put(5476,-886){\makebox(0,0)[lb]{\smash{\SetFigFont{12}{14.4}{\familydefault}{\mddefault}{\updefault}$q_2$}}}
\put(5851,-1936){\makebox(0,0)[lb]{\smash{\SetFigFont{12}{14.4}{\familydefault}{\mddefault}{\updefault}$2^{-1}$}}}
\put(5476,-2986){\makebox(0,0)[lb]{\smash{\SetFigFont{12}{14.4}{\familydefault}{\mddefault}{\updefault}$q_3$}}}
\put(4651,-3211){\makebox(0,0)[lb]{\smash{\SetFigFont{12}{14.4}{\familydefault}{\mddefault}{\updefault}$3^{-1}$}}}
\put(676,-586){\makebox(0,0)[lb]{\smash{\SetFigFont{12}{14.4}{\familydefault}{\mddefault}{\updefault}$1^{-1}$}}}
\put(2326,-3211){\makebox(0,0)[lb]{\smash{\SetFigFont{12}{14.4}{\familydefault}{\mddefault}{\updefault}$(n-3)^{-1}$}}}
\put(1576,-2986){\makebox(0,0)[lb]{\smash{\SetFigFont{12}{14.4}{\familydefault}{\mddefault}{\updefault}$q_{n-2}$}}}
\put(1876,-1486){\makebox(0,0)[lb]{\smash{\SetFigFont{12}{14.4}{\familydefault}{\mddefault}{\updefault}$n-2$}}}
\put(451,-1486){\makebox(0,0)[lb]{\smash{\SetFigFont{12}{14.4}{\familydefault}{\mddefault}{\updefault}$q_{n-1}$}}}
\put(301,-2536){\makebox(0,0)[lb]{\smash{\SetFigFont{12}{14.4}{\familydefault}{\mddefault}{\updefault}$(n-2)^{-1}$}}}
\end{picture}

%% file: cycle.pstex_t
\begin{picture}(0,0)%
\epsfig{file=cycle.pstex}%
\end{picture}%
\setlength{\unitlength}{3947sp}%
\begingroup\makeatletter\ifx\SetFigFont\undefined%
\gdef\SetFigFont#1#2#3#4#5{%
  \reset@font\fontsize{#1}{#2pt}%
  \fontfamily{#3}\fontseries{#4}\fontshape{#5}%
  \selectfont}%
\fi\endgroup%
\begin{picture}(6320,2648)(214,-1944)
\put(1801,-361){\makebox(0,0)[lb]{\smash{\SetFigFont{12}{14.4}{\rmdefault}{\mddefault}{\updefault}1}}}
\put(2851,-361){\makebox(0,0)[lb]{\smash{\SetFigFont{12}{14.4}{\rmdefault}{\mddefault}{\updefault}1}}}
\put(4051,-361){\makebox(0,0)[lb]{\smash{\SetFigFont{12}{14.4}{\rmdefault}{\mddefault}{\updefault}$1^{-1}$}}}
\put(4876,-361){\makebox(0,0)[lb]{\smash{\SetFigFont{12}{14.4}{\rmdefault}{\mddefault}{\updefault}$1^{-1}$}}}
\put(4801,539){\makebox(0,0)[lb]{\smash{\SetFigFont{12}{14.4}{\rmdefault}{\mddefault}{\updefault}$1^{-1}$}}}
\put(5776,-961){\makebox(0,0)[lb]{\smash{\SetFigFont{12}{14.4}{\rmdefault}{\mddefault}{\updefault}$1^{-1}$}}}
\put(5401,-1786){\makebox(0,0)[lb]{\smash{\SetFigFont{12}{14.4}{\rmdefault}{\mddefault}{\updefault}$1^{-1}$}}}
\put(4576,-1786){\makebox(0,0)[lb]{\smash{\SetFigFont{12}{14.4}{\rmdefault}{\mddefault}{\updefault}$1^{-1}$}}}
\put(1576,-1786){\makebox(0,0)[lb]{\smash{\SetFigFont{12}{14.4}{\rmdefault}{\mddefault}{\updefault}1}}}
\put(2701,-1786){\makebox(0,0)[lb]{\smash{\SetFigFont{12}{14.4}{\rmdefault}{\mddefault}{\updefault}1}}}
\put(3451,-1786){\makebox(0,0)[lb]{\smash{\SetFigFont{12}{14.4}{\rmdefault}{\mddefault}{\updefault}1}}}
\put(1276,-211){\makebox(0,0)[lb]{\smash{\SetFigFont{12}{14.4}{\rmdefault}{\mddefault}{\updefault}$q_0$}}}
\put(5551,-211){\makebox(0,0)[lb]{\smash{\SetFigFont{12}{14.4}{\rmdefault}{\mddefault}{\updefault}$q_{n-1}$}}}
\put(6001,-1636){\makebox(0,0)[lb]{\smash{\SetFigFont{12}{14.4}{\rmdefault}{\mddefault}{\updefault}$q_{n-1}$}}}
\put(3976,-1636){\makebox(0,0)[lb]{\smash{\SetFigFont{12}{14.4}{\rmdefault}{\mddefault}{\updefault}$q_{\frac{n}{2}}$}}}
\put(2026,-1636){\makebox(0,0)[lb]{\smash{\SetFigFont{12}{14.4}{\rmdefault}{\mddefault}{\updefault}$q_1$}}}
\put(976,-1636){\makebox(0,0)[lb]{\smash{\SetFigFont{12}{14.4}{\rmdefault}{\mddefault}{\updefault}$q_0$}}}
\put(3451,-211){\makebox(0,0)[lb]{\smash{\SetFigFont{12}{14.4}{\rmdefault}{\mddefault}{\updefault}$q_{\frac{n-1}{2}}$}}}
\end{picture}